\documentclass{article}
\usepackage{arxiv}
\usepackage{palatino}
\usepackage{microtype}
\usepackage{mathtools}
\usepackage{booktabs}
\usepackage[numbers,sort&compress]{natbib}
\usepackage{graphicx}
\usepackage{amsmath,amssymb,amsfonts,amsxtra,bm}
\usepackage{amsthm}
\usepackage{xspace}
\usepackage{enumerate}
\usepackage{hyperref}

\usepackage{xcolor}
\definecolor{darkblue}{rgb}{0,0.0,0.55}
\hypersetup{ 
  colorlinks = true,
  citecolor  = darkblue,
  linkcolor  = darkblue,
  citecolor  = darkblue,
  filecolor  = darkblue,
  urlcolor   = darkblue,
}

\usepackage{url}

\newcommand{\dist}{\delta}

\newcommand{\frob}[1]{\|#1\|_{\text{F}}}
\renewcommand{\H}{\mathbb{H}}
\newcommand{\pd}{\mathbb{P}}
\newcommand{\reals}{\mathbb{R}}
\newcommand{\half}{\tfrac12}
\newcommand{\pfrac}[2]{\left(\tfrac{#1}{#2}\right)}
\newcommand{\ld}{\ell{}d}
\newcommand{\ip}[2]{\langle {#1},\, {#2} \rangle}

\DeclareMathOperator{\trace}{tr}
\DeclareMathOperator{\qjsd}{QJSD}
\DeclareMathOperator{\qjrd}{QJRD}
\DeclareMathOperator*{\argmin}{argmin}

\newtheorem{theorem}{Theorem}[section]

\newtheorem{cor}[theorem]{Corollary}

\newtheorem{conj}[theorem]{Conjecture}

\newtheorem{prob}[theorem]{Problem}

\theoremstyle{definition}
\newtheorem{defn}[theorem]{Definition}
\newtheorem{example}[theorem]{Example}
\newtheorem{rmk}[theorem]{Remark}

\numberwithin{equation}{section}

\begin{document}
\title{Metrics Induced by Jensen-Shannon and Related Divergences on Positive Definite Matrices}

\author{\name{Suvrit Sra} \email{suvrit@mit.edu}\\
  \addr{Massachusetts Institute of Technology, Cambridge, MA 02139}
}

\maketitle
\begin{abstract}
  We study metric properties of symmetric divergences on Hermitian positive definite matrices. In particular, we prove that the square root of these divergences is a distance metric. As a corollary we obtain a proof of the metric property for Quantum Jensen-Shannon-(Tsallis) divergences (parameterized by $\alpha\in [0,2]$), which in turn (for $\alpha=1$) yields a proof of the metric property of the Quantum Jensen-Shannon divergence that was conjectured by 
Lamberti \emph{et al.} a decade ago (\emph{Metric character of the quantum Jensen-Shannon divergence}, Phy.\ Rev.\ A, \textbf{79}, (2008).) A somewhat more intricate argument also establishes metric properties of Jensen-R\'enyi divergences (for $\alpha \in (0,1)$), and outlines a technique that may be of independent interest. 
\end{abstract}


\section{Introduction}
We study metric properties of symmetrized divergence measures on hermitian positive definite (hpd) matrices. Such divergence measures are used in a wide variety of applications, ranging from quantum information theory~\cite{briet09,lamberti2008metric,Dajka2011}, to optimization~\cite{sra2016matrix,sra2015conic}, to machine learning and computer vision~\cite{wang2015beyond,zhang2015learning,cherian2012jensen}, among others~\cite{yger2016riemannian,harandi2017dimensionality,rossi2013characterizing,lin1991divergence,fuglede2004jensen,briet09,burbea1982convexity,burbea1982entropy,sdiv,nielsen2011burbea}.

Our focus on studying metric properties of these divergences is primarily driven by the aim to build a theory that answers Conjecture~\ref{conj:one} as a special case. A secondary aim is to obtain a family of metrics closely related to the S-Divergence~\cite{sdiv} (which has found a variety of applications), and thus provide a new family of potentially useful metrics. Remarkably, the metric property of the S-Divergence, which was the central result of~\cite{sdiv}, plays a crucial role in the present paper too.

The divergence underlying our primary aim is obtained by symmetrizing a Bregman divergence (see Section~\ref{sec:background}), or equivalently, by using midpoint convexity. For instance, consider the \emph{von Neumann entropy}
\begin{equation}
  \label{eq:vnent}
  S(X) := -\trace(X \log X), \qquad X \in \pd_d,
\end{equation}
which leads to the so-called the \emph{Quantum Jensen-Shannon divergence}~\cite{lamberti2008metric}:
\begin{align}
  \label{eq:18}
  \qjsd(X,Y) &:= S\pfrac{X+Y}{2}-\half\bigl(S(X)+S(Y)\bigr).
\end{align}
Divergence~\eqref{eq:18} has found a variety of applications, including several cited above. While it is clearly symmetric and nonnegative, it is not a true distance; nevertheless, empirically its square root $\qjsd^{1/2}$ has been observed to satisfy the triangle inequality~\cite{lamberti2008metric,briet09}.

A formal study $\qjsd^{1/2}$ as a metric was started by Lamberti et al.~\cite{lamberti2008metric}, who used it for measuring distances between quantum states. They also showed that $\qjsd$ is the square of a metric for pure states. Shortly thereafter, Bri\"et and Harremo\"es~\cite{briet09} claimed  that~\eqref{eq:19} is the square of a Hilbertian metric for qubits and pure states of any dimension; their proof, apparently contains an error, and a proof was furnished by Carlen, Lieb and Seiringer---please see \cite[\S3]{daniel} for more details. However, for general quantum states no results are known (apart from the parallel work~\cite{daniel}, which was brought to our notice by its author, and which furnishes a proof of Conjecture~\ref{conj:one}). Specifically, Lamberti et al.\ made the following conjecture:

\vskip5pt
\noindent\colorbox{gray!20}
{\begin{minipage}{.97\linewidth}
  \begin{conj}[Lamberti et al.~\cite{lamberti2008metric}]
  \label{conj:one}
  $\qjsd^{1/2}$ is a metric on $\pd_d$ (see also \cite{briet09}).
\end{conj}
\end{minipage}
}

\vskip4pt
\noindent We address this conjecture as a byproduct of the following more general task:
\begin{center}
  \em Provide simple conditions on Jensen divergences to be squares of metrics.
\end{center}

\subsection{Summary}
The main contributions of this paper are as follows:
\begin{list}{{\tiny\color{darkblue}$\blacksquare$}}{\leftmargin=3em}
  \setlength{\itemsep}{3pt}
\item We prove in Theorem~\ref{thm:one} the metric property for a rich class of Jensen divergences (please see Section~\ref{sec:background} for background). This class includes  $\qjsd_\alpha$ ($\alpha \in [0,2]$) as a special case, and thus \emph{a fortiori} also includes $\qjsd$, answering Conjecture~\ref{conj:one}. Moreover, it extends to more general settings based on certain convex functions (Theorem~\ref{thm:two}). Both Theorem~\ref{thm:one} and~\ref{thm:two} rely on integrals related to Pick functions (Section~\ref{sec:pick}).
  
\item We subsequently prove in Theorem~\ref{thm:renyi} the harder result that the quantum Jensen-R\'enyi divergence $\qjrd_\alpha$ is the square of a metric for $\alpha \in (0,1)$. Our technique relies on integral representations of completely monotonic functions and an argument based on $3\times 3$ matrices that may be of independent interest.

\item In Section~\ref{sec:relative} we prove the Jensen-Shannon divergence generated by the $\alpha$-Tsallis relative entropy is also the square of a metric. 
  
\end{list}

\noindent\textbf{Note:} The closely related work~\cite{daniel} appeared parallel to ours.\footnote{More pedantically, that work appeared on arXiv a few days before our work appeared there.} The main result of~\cite{daniel} is a proof of the metric property of $\qjsd$. Our results cover not only $\qjsd$, but identify a general integral representation via a connection to Pick functions, and establish metric properties for a much wider class of divergences.

\section{Background}
\label{sec:background}
We begin by recalling some basic facts about divergences. Perhaps the most well-known divergence is the \emph{Bregman divergence}~\cite{censor97}\footnote{Bregman divergences over scalars and vectors have been well-studied; see e.g.,~\cite{censor97,banerjee04b}. They are called divergences because they are not distances.}, which is generated by differentiable and convex function $f : \reals^n \to \reals$ as follows,
\begin{equation}
  \label{matrix:eq.breg.scalar}
  D_f(x,y) := f(x) - f(y) - \ip{\nabla f(y)}{x-y}.
\end{equation}

By construction, $D_f(x,y)$ is nonnegative, convex in $x$, and equals $0$ if $x=y$. It is typically asymmetric and does not satisfy the triangle inequality, which explains the name ``divergence.''

\begin{example}
  \label{eg.bdiv}
  Some common Bregman divergences are listed below.
  \begin{itemize}
    \setlength{\itemsep}{-1pt}
  \item \textbf{Squared $\ell_2$-distance}: Let $f(x)=\half x^Tx$, then $D_f(x,y)=\half\|x-y\|_2^2$.
  \item \textbf{KL divergence} on $\reals_{++}$. $f(x) = x \log x$, so $D_f(x,y)=x\log(x/y)-x+y$.
  \item \textbf{Burg divergence} on $\reals_{++}$. $f(x)=-\log x$, so $D_f(x,y)=\log(y/x)+x/y-1$.
  \end{itemize}
\end{example}

The Bregman divergence~\eqref{matrix:eq.breg.scalar} extends naturally to hermitian matrices. Let $X, Y$ be hermitian, and let the scalar function $f$ be defined on hermitian matrices the usual way (via spectral decomposition), then the \emph{Bregman matrix divergence} is defined as
\begin{equation}
  \label{matrix:eq.breg}
  D_f(X,Y) := \trace f(X) -\trace f(Y) - \ip{f'(Y)}{X-Y}.
\end{equation}
It is an instructive exercise to verify that $D_f(X,Y) \ge 0$.
\begin{example}
  \label{eg.mdiv}
  The matrix versions of Example~\ref{eg.bdiv} are:
  \begin{itemize}
    \setlength{\itemsep}{-1pt}
  \item \textbf{Squared Frobenius distance}: here $\trace f(X)=\trace(X^2)$, so that $D_f(X,Y)=\half\frob{X-Y}^2$.
  \item \textbf{von Neumann divergence (Umegaki relative entropy)}: here $\trace f(X)=\trace(X\log X)$, so $D_f(X,Y)$ yields the von Neumann divergence of quantum information theory~\cite{nielsen2002quantum}:
  \begin{equation*}
    D_{\text{vN}}(X,Y) = \trace(X\log X - X\log Y-X+Y).
  \end{equation*}
  \item \textbf{Stein's loss}: here $\trace f(X)=-\log\det(X)$, so $D_f(X,Y)$ becomes
  \begin{equation*}
    D_{\ld}(X,Y) = \trace(Y^{-1}(X-Y)) - \log\det(XY^{-1}),
  \end{equation*}
  which is also known as the \emph{LogDet Divergence}~\cite{kulis}, or \emph{Stein's loss}~\cite{stein56}.
  \end{itemize}
\end{example}

\subsection{Jensen and Jensen-Shannon Divergences}
Although Bregman divergences are widely useful, their asymmetry can be undesirable. A popular symmetric alternative is the \emph{Jensen divergence} (sometimes called Jensen-Bregman divergence~\cite{cherian2012jensen}):
\begin{equation}
  \label{matrix:eq.js}
  S_f(X,Y) := \half\bigl(D_f(X, \tfrac{X+Y}{2}) + D_f(Y, \tfrac{X+Y}{2})\bigr).
\end{equation}
This divergence has two possibly more transparent representations:
\begin{align}
  \label{matrix:eq.js1}
  S_f(X,Y) &\quad=\quad\half\bigl[\trace f(X) + \trace f(Y)\bigr] - \trace f\bigl(\tfrac{X+Y}{2}\bigr),\\
  \label{matrix:eq.js2}
  S_f(X,Y) &\quad=\quad \min_Z\quad \half[D_f(X,Z) + D_f(Y,Z)].
\end{align}
\begin{rmk}
  In some contexts~\eqref{matrix:eq.js1} is also called the \emph{Jensen-Shannon divergence}. But for clarity within the context of the quantum setting, we reserve that name for symmetrization~\eqref{matrix:eq.js2} applied to $D_f$ being a suitable quantum relative entropy.
\end{rmk}

\begin{example}
  The symmetric versions of Example~\ref{eg.mdiv} are:
  \begin{itemize}
    \setlength{\itemsep}{-1pt}
  \item If $\trace f(X)=\half\trace X^2$, we obtain $S_f(X,Y)=D_f(X,Y)=\half\frob{X-Y}^2$.
  \item If $\trace f(X)=\trace(X\log X)$, both~\eqref{matrix:eq.js1} and~\eqref{matrix:eq.js2} yield the $\qjsd$~(\ref{eq:18}).
  %
  \item For $\trace f(X)=-\log\det(X) \equiv -\ld(X)$, we obtain the \emph{S-Divergence}~\cite{sdiv}:
    \begin{equation}
      \label{matrix:eq.ss}
      S_f(X,Y) \equiv S_{\ld}(X,Y) := \delta_S^2(X,Y) := \ld\pfrac{X+Y}{2}-\half\ld(X)-\half\ld(Y).
    \end{equation}
  \end{itemize}
\end{example}
With this background, we are now ready to present the main results on this paper.

\section{Metric Properties of Quantum Jensen Divergences}
In this section we study symmetric divergence whose square roots are metrics. The class of divergences we cover is chosen to be able to capture the $\alpha$-Tsallis generalization to QJSD~(\ref{eq:18}):
\begin{align}
  \label{eq:19}
  \qjsd_{\alpha}(X,Y) &:= S_\alpha\pfrac{X+Y}{2}-\half\bigl(S_\alpha(X)+S_\alpha(Y)\bigr),
\end{align}
where $S_\alpha$ is the \emph{$\alpha$-Tsallis entropy}
\begin{align}
  \label{eq:tsallis}
  S_\alpha(X) &:= \frac{\trace(X^\alpha)-\trace X}{1-\alpha},\quad \alpha\in [0,2]\setminus\{1\}.
\end{align}
Note that, $\lim_{\alpha\to 1^+}S_\alpha(X)=S(X)$; and that $S_\alpha$ is concave on hpd matrices. 

We present a technique which implies that $\qjsd_\alpha^{1/2}$ is a metric as a special case; this result in turns immediately solves Conjecture~\ref{conj:one} (which corresponds to $\alpha\to 1^{+}$). Specifically, consider a function $f$ that admits the following representation on $(0,\infty)$:
\begin{equation}
  \label{eq:7}
  f(x) = a + b x + c \log x + \int_{0}^\infty \log\frac{(t+x)}{h(t)}d\mu(t),
\end{equation}
where $a, b \in \reals$, $c \ge 0$, $h(t)>0$, and $\mu$ is a nonnegative measure. This function is concave, so using it we define the \emph{quantum Jensen divergence} (on hpd matrices):
\begin{equation}
  \label{eq:8}
  \Delta_f(X,Y) := \trace\bigl[f\bigl(\tfrac{X+Y}{2}\bigr)-\half f(X)-\half f(Y)\bigr],
\end{equation}
Our first main result is Theorem~\ref{thm:one}.

\vskip5pt
\noindent\colorbox{gray!20}{\begin{minipage}{0.97\linewidth}
\begin{theorem}
  \label{thm:one}
  Let $f$ be a function that admits the representation~\eqref{eq:7}, and let $\Delta_f$ be the Jensen-divergence~\eqref{eq:8}. Then, $\Delta_f^{1/2}$ is a distance on $\pd_d$.
\end{theorem}
\end{minipage}
}

\vskip5pt
\noindent Crucial to our proof is the metric property of the S-Divergence~\eqref{matrix:eq.ss}. 

\vskip4pt
\noindent\colorbox{gray!20}{\begin{minipage}{0.97\linewidth}
\begin{theorem}[Sra~(2016)~\cite{sdiv}]
  \label{thm:sdiv}
  $\delta_S$ given by~\eqref{matrix:eq.ss} is a metric.
\end{theorem}
\end{minipage}}

\vskip5pt
\begin{proof}[Proof of Theorem~\ref{thm:one}]
  The only non-trivial part is to prove the triangle inequality for $\Delta_f^{1/2}$.
  Using~\eqref{eq:7} and noting that $\trace\log(X) = \ld(X)$
  we can express $\Delta_f$ as
  \begin{equation*}
    \begin{split}
      &\Delta_f(X,Y) = c\bigl(\ld\pfrac{X+Y}{2}-\half\ld(X)-\half\ld(Y)\bigr)\\
      &+\int_{0}^\infty\bigl[\ld(tI+\tfrac{X+Y}{2})-\half\ld((tI+X)(tI+Y))\bigr]d\mu(t),
    \end{split}
  \end{equation*}
  which may be written in terms of the S-Divergence as
  \begin{equation}
    \label{eq:9}
    \Delta_f(X,Y) = c\dist_S^2(X,Y)+\int_{0}^\infty\dist_S^2(tI+X, tI+Y)d\mu(t).
  \end{equation}
  Since $c\ge 0$, it follows from~\eqref{eq:9} that $\Delta_f$ is a non-negatively weighted sum of squared distances, therefore $\Delta_f^{1/2}$ satisfies the triangle inequality, completing the proof.
\end{proof}

\vskip4pt
\noindent\colorbox{gray!20}
{\begin{minipage}{0.97\linewidth}
\begin{cor}
  \label{cor:1}
  Let $\alpha \in (0,1)$. Then, $\qjsd_\alpha^{1/2}$ is a metric.
\end{cor}
\end{minipage}
}
\begin{proof}
  For $0<\alpha<1$ and $x>0$, we use the representation (also exploited in~\cite{sra2019}):
  \begin{equation}
    \label{eq:10}
    x^\alpha = \frac{\alpha \sin(\alpha \pi)}{\pi}\int_{0}^\infty\log\pfrac{t+x}{t}t^{\alpha-1}dt,
  \end{equation}
  which is an instance of~\eqref{eq:7}. Now, Theorem~\ref{thm:one} immediately yields the corollary.
\end{proof}

While Corollary~\ref{cor:1} holds for $\alpha \in (0,1)$, a slightly different integral representation allows us to also obtain the following result.

\vskip4pt
\noindent\colorbox{gray!20}
{\begin{minipage}{0.97\linewidth}
\begin{cor}
  \label{cor:2}
  Let $\alpha \in (1,2)$. Then, $\qjsd_\alpha^{1/2}$ is a metric.
\end{cor}
\end{minipage}}
\begin{proof}
  The key idea is to use the following integral representation (for $1<\alpha < 2$):
  \begin{equation}
    \label{eq:11}
    x^\alpha = \frac{|\alpha \sin(\alpha \pi)|}{\pi}\int_{0}^\infty(tx - \log(1+tx))t^{-\alpha-1}dt,
  \end{equation}
  which was noted by~\cite{sra2019}; notice that this representation is not captured by~\eqref{eq:7}. Using~\eqref{eq:11} and arguing as for Theorem~\ref{thm:one}, the proof readily follows.
\end{proof}
Observing that $\lim_{\alpha\to 1^+}\qjsd_\alpha=\qjsd$, we obtain a proof for Conjecture~\ref{conj:one}. 

\vskip4pt
\noindent\colorbox{gray!20}{\begin{minipage}{0.97\linewidth}
\begin{cor}
  \label{cor:3}
  $\qjsd^{1/2}$ is a metric on $\pd_d$ (i.e., Conjecture~\ref{conj:one} is true).
\end{cor}
\end{minipage}}

\vskip4pt
\begin{rmk}
  In parallel work, Corollary~\ref{cor:3} appeared as the main result of~\cite{daniel}. We are grateful to D.~Virosztek for bringing his work~\cite{daniel} to our attention.
\end{rmk}

\subsection{Generalizing Corollary~\ref{cor:2}}
The reader has perhaps already realized that the argument used to prove Corollary~\ref{cor:2} also holds for convex functions on $(0,\infty)$ that admit the representation
\begin{equation}
  \label{eq:12}
  f(x) = a + bx - c \log x + \int_{0}^\infty(tx-\log(1+tx))d\mu(t),
\end{equation}
where $a, b \in \reals$, $c \ge 0$, and $\mu$ is a nonnegative measure. Then we have the following:

\vskip4pt
\noindent\colorbox{gray!20}{\begin{minipage}{0.97\linewidth}
\begin{theorem}
  \label{thm:two}
  Let $f(x)$ be given by~\eqref{eq:12}, and define the Jensen divergence
  \begin{equation}
    \label{eq:13}
    \Delta_f(X,Y) := \half\trace f(X)+\half\trace f(Y) - \trace f\pfrac{X+Y}{2}.
  \end{equation}
  Then, $\Delta^{1/2}$ is a distance function on $\pd_d$.
\end{theorem}
\end{minipage}}

\section{The Jensen-Shannon $\alpha$-Tsallis relative entropy}
\label{sec:relative}
This section addresses a question posed by a reader of an earlier version of this paper. They noted that the channel capacity interpretation of the \emph{Jensen-Shannon Tsallis relative entropy}~\eqref{eq:28} makes studying its corresponding metric property more valuable than that of $\qjsd_\alpha$. However, as shown below, this property easily follows from that of $\qjsd_\alpha$.

In particular, consider first the \emph{$\alpha$-Tsallis relative entropy}
\begin{equation*}
  S_\alpha(X,Y) := \frac{\alpha\trace X + (1-\alpha)\trace Y - \trace X^{\alpha}Y^{1-\alpha}}{1-\alpha},\quad \alpha \in (0,\infty)\setminus\{1\}.
\end{equation*}
The ``centroid'' of its symmetrization reduces to a power-mean; more precisely,
\begin{equation}
  \label{eq:3}
  \argmin_{Z \in \pd_d}\quad S_\alpha(X,Z)+ S_\alpha(Y,Z) = \pfrac{X^\alpha + Y^\alpha}{2}^{1/\alpha}.
\end{equation}
Using~\eqref{eq:3}, we arrive at the main result of this section.

\vskip4pt
\noindent\colorbox{gray!20}{\begin{minipage}{0.97\linewidth}
\begin{theorem}
  \label{thm:tsallis.relative}
  Let the Jensen-Shannon $\alpha$-Tsallis divergence be defined as
  \begin{equation}
    \label{eq:28}
    \begin{split}
      \Delta_\alpha(X,Y) &:= S_\alpha(X,M)+ S_\alpha(Y,M)\\
      &= \tfrac{1}{1-\alpha}\left(\alpha\trace(X+Y)+2(1-\alpha)\trace Z - \trace(Z^{1-\alpha}(X^\alpha+Y^\alpha))\right),
    \end{split}
  \end{equation}
  where $M$ denotes the rhs of~\eqref{eq:3}. Then, for $\alpha \ge\half$, $\Delta_\alpha^{1/2}$ is a metric.
\end{theorem}
\end{minipage}}

\begin{proof} We prove this metricity by reducing it to Corollaries~\ref{cor:1} and~\ref{cor:2}. To that end, write $A = X^\alpha$ and $B^\alpha$, and $t=1/\alpha$. Then, we see that 
  \begin{equation}
    \label{eq:29}
    \Delta_\alpha(X,Y) = \tfrac{1}{1-\alpha}\left(\alpha\trace(A^t+B^t) - 2\alpha\trace\pfrac{A+B}{2}^t\right).
  \end{equation}
  For $\alpha \in [\half, 1)$, we have $t \in (1,2]$, whereby we can apply Corollary~\ref{cor:2} to deduce metricity of $\Delta_\alpha^{1/2}$. For the case $\alpha > 1$, we have $t < 1$, and also $1-\alpha < 0$. In this case, we can apply Corollary~\ref{cor:1} to obtain metricity of $\Delta_\alpha^{1/2}$.
\end{proof}

\begin{rmk}
  Similarly, one could consider Jensen-Shannon versions of R\'enyi relative entropies such as the Petz-R\'enyi relative entropy~\cite{petz1986quasi} and the sandwiched R\'enyi relative entropy~\cite{muller13}; we defer such a discussion to the future.
\end{rmk}

\section{Divergences and Pick functions}
\label{sec:pick}
We briefly remark below on the deeper connection that motivated our choice~\eqref{eq:7}. This connection also provides a valuable converse, namely, conditions on when such a representation holds for a given function. 

In particular, in~\cite[Theorem~2.1]{silhavy15} it was shown that if $xf'(x)$ has an analytic extension whose restriction to the upper half plane is a Pick function~\cite{donoghue} and $xf'(x)$ is bounded, then $f$ admits the representation (valid for $x>0$):
\begin{equation}
  \label{eq:14}
    f(x) = a + bx + c\log x + \frac{d}{x}+ \int_{0}^\infty\Bigl[\log\frac{(t+x)}{(1+t)}-\frac{\log x}{1+t^2}\Bigr]d\mu(t),
\end{equation}
with $a, c \in \reals$, $b, d \ge 0$, and the nonengative measure $\mu$ satisfies $\int_0^\infty t/(1+t^2)d\mu(t) < \infty$. Moreover, if in addition $f'\ge 0$, then (see \cite[Thm.~4.2]{silhavy15})
\begin{equation}
  \label{eq:15}
  f(x) = a + bx + c\log x + \int_{0}^\infty\log\frac{t+x}{1+t}d\mu(t),
\end{equation}
with $a\in \reals$, $b, c\ge 0$ and $\int_0^\infty (1+t^2)^{-1}d\mu(t)< \infty$; this form is what motivates our slightly more general choice~\eqref{eq:7}. 

\section{Quantum Jensen-R\'enyi Divergence}
Recall that the Quantum R\'enyi Entropy is defined as
\begin{equation}
  \label{eq:22}
  H_\alpha(X) := \frac{1}{1-\alpha}\log\frac{\trace(X^\alpha)}{\trace(X)},\quad \alpha \ge 0, \alpha \not= 1.
\end{equation}
Observe that~\eqref{eq:22} is concave for $\alpha \in (0,1)$; thus, for such $\alpha$ we can define the \emph{Quantum Jensen-R\'enyi Divergence} as:
\begin{equation}
  \label{eq:17}
  \qjrd_\alpha := H_\alpha\pfrac{X+Y}{2} - \half H_\alpha(X) - \half H_\alpha(Y).
\end{equation}

Proving that $\qjrd_\alpha$ is the square of a metric (Theorem~\ref{thm:renyi}) turns out to be harder than analyzing $\qjsd_\alpha$. Indeed, $\qjrd_\alpha$ is directly amenable to the Pick function technique developed above, and it requires a more intricate argument based on two more ingredients: \emph{complete monotonicity} and the relation between \emph{conditionally negative definite} matrices and metrics.

\begin{defn}[Complete monotonicity]
  A function $F: (0,\infty) \to \reals$ is called \emph{completely monotonic (CM)} if $(-1)^kF^{(k)}(x) \ge 0$, for $k\ge0$. Bernstein's theorem (see e.g.,~\cite[Thm.~6.13]{berg84}) shows that such an $F$ can be written as
  \begin{equation}
    \label{eq:23}
    F(x) = \int_0^\infty e^{-t x}d\nu(t),
  \end{equation}
  for a nonnegative measure $\nu$ on $[0,\infty)$.
\end{defn}
\begin{defn}[CND]
  $X \in \H_d$ is \emph{conditionally negative definite (cnd)} if
  \begin{equation}
    \label{eq:20}
    v^*Xv \le 0,\quad\text{for all } v\in\mathbb{C}^d\ \text{s.t. } v^*1 = 0.
  \end{equation}
\end{defn}

\vskip4pt
\noindent\colorbox{gray!20}{\begin{minipage}{0.97\linewidth}
\begin{theorem}
  \label{thm:renyi}
  $\qjrd_\alpha^{1/2}$ is a metric on HPD matrices for $\alpha \in (0,1)$.
\end{theorem}
\end{minipage}}

\begin{proof}
  Without loss of generality we may assume that $\trace(X)=1$. Introduce now the shorthand $d_{xy} = \trace\pfrac{X+Y}{2}^\alpha$, $d_x=d_{xx}$; define $d_{xz}, d_{yz}$, and $d_y, d_z$ similarly. Then, by Theorems~\ref{thm:one} and \ref{thm:cnd} it follows that the matrix
  \begin{equation*}
    D =
    \begin{bmatrix}
      0              & 2d_{xy}-d_x-d_y & 2d_{xz}-d_x-d_z\\
      2d_{xy}-d_x-d_y & 0              & 2d_{yz}-d_y-d_z\\
      2d_{xz}-d_x-d_z & 2d_{yz}-d_y-d_z & 0
    \end{bmatrix},
  \end{equation*}
  is cnd. It is also known that (see e.g., \cite[Thm.~4.4.2]{bapat1997nonnegative}) that if an elementwise nonnegative matrix $[m_{ij}]$ is cnd, and $F$ is a CM function, then $[F(m_{ij})]$ is positive definite. Let $\bm{\theta}=[d_x, d_y, d_z]^T$; then, $M=\bm{\theta}\bm{1}^T+\bm{1}\bm{\theta}^T + D$ is also cnd, and so is $2t\bm{11}^T+M$ for all $t \ge 0$. Thus, using the CM function $F(s)=2/s$ on $M$ we see that
  \begin{equation}
    \label{eq:24}
    M' =
    \begin{bmatrix}
      \frac{1}{t+d_x} & \frac{1}{t+d_{xy}} & \frac{1}{t+d_{xz}}\\
      \frac{1}{t+d_{xy}} & \frac{1}{t+d_{y}} & \frac{1}{t+d_{yz}}\\
      \frac{1}{t+d_{xz}} & \frac{1}{t+d_{yz}} & \frac{1}{t+d_{z}}
    \end{bmatrix} \succeq 0.
  \end{equation}
  Using $\bm{\eta}=\frac12[t+d_x, t+d_y, t+d_z]^T$, we construct $\bm{\eta1}^T+\bm{1\eta}^T-M'$, which is clearly cnd. Explicitly, this matrix is given by (we suppress symmetric entries via $*$ for brevity):
  \begin{equation*}
    \begin{bmatrix}
      0 & \frac{1}{2}\bigl(\tfrac{1}{t+d_{x}}+\tfrac{1}{t+d_{y}}\bigr) -\frac{1}{t+d_{xy}} & \frac{1}{2}\bigl(\tfrac{1}{t+d_{x}}+\tfrac{1}{t+d_{z}}\bigr)-\frac{1}{t+d_{xz}}\\
      * & 0 & \frac{1}{2}\bigl(\tfrac{1}{t+d_{y}}+\tfrac{1}{t+d_{z}}\bigr)-\frac{1}{t+d_{yz}}\\
      * & * & 0
    \end{bmatrix}.
  \end{equation*}
  Since $d_x=\trace(X^\alpha)$ is concave, it follows that $1/(t+d_x)$ is convex, whereby $$\frac{1}{2}\Bigl(\frac{1}{t+d_{x}}+\frac{1}{t+d_{y}}\Bigr) -\frac{1}{t+d_{xy}} \ge 0.$$ Thus, we can invoke Theorem~\ref{thm:cnd} again to conclude that
  \begin{equation}
    \label{eq:2}
    \frac{1}{2}\Bigl(\frac{1}{t+d_{x}}+\frac{1}{t+d_{y}}\Bigr) -\frac{1}{t+d_{xy}} =: \delta_t^2(X,Y),
  \end{equation}
  where $\delta_t$ is a distance metric. Next, recall the following integral representation
  \begin{equation}
    \label{eq:1}
    \log x = -\int_{0}^\infty\Bigl(\frac{1}{t+x}-\frac{t}{1+t^2}\Bigr)dt,
  \end{equation}
  which can be obtained for instance by first writing $(\log x)^2$ using the representation~(\ref{eq:14}) and then differentiating~\cite{silhavy15}. Integrating~\eqref{eq:2} using representation~\eqref{eq:1} we can finally write
  \begin{equation*}
    \qjrd_\alpha(X,Y) = \int_0^\infty \delta_t^2(X,Y)dt,
  \end{equation*}
  which proves that $\qjrd_\alpha^{1/2}$ is a metric.
\end{proof}

\subsection{Other Extensions} The above proof actually also shows that if $d_{xy}=h\pfrac{X+Y}{2}$ where $h(X)$ is concave and $\Delta_h(X,Y) = d_{x,y}-\half d_x - \half d_y$ is the square of a metric, then
\begin{equation*}
  \Delta_F(X,Y) := \half F(d_x)+ \half F(d_y) - F(d_{xy}),
\end{equation*}
is the square of a metric. Indeed, if $h$ is concave, then $e^{-t h}$ is convex for $t\ge 0$. Thus, for a CM function $F$ the map $F(h(X))$ is convex, whence $\Delta_F(X,Y) \ge 0$. The triangle inequality follows from a construction analogous to~\eqref{eq:24}. We omit details for brevity.

\section{Conclusions}
In this paper, we identified sufficient conditions based on Pick-Nevanlinna integral representations for ensuring that the corresponding (quantum) Jensen divergence is the square of a metric. At this point, it is natural to consider the following (likely harder) task as an open problem:
\begin{prob}
  Identify conditions that are both necessary and sufficient for a given (quantum) Jensen divergence to be the square of a metric.
\end{prob}

\subsection*{Acknowledgments}
I would like to thank D\'aniel Virosztek for bringing~\cite{daniel} to my attention, and also for pointing out to me certain corrections in the attribution of the QJSD conjecture, as well as the reference to a proof by Carlen-Lieb-Seiringer (see \cite[\S3]{daniel}) as the correct proof of metricity for qubits. I also thank an anonymous referee of a submitted version of this paper for suggesting the idea that led to the discussion in Section~\ref{sec:relative}.

\bibliographystyle{plainnat}

\appendix
\section{Distances and $3\times 3$ cnd matrices}
In this section, we summarize the equivalence between $3\times 3$ cnd matrices and corresponding distance metrics. This material is classical, but we include our own proofs for completeness. Indeed, squared distances are intimately related with cnd matrices. The following theorem summarizes this connection.

\begin{theorem}
  \label{thm:cnd}
  Let $D$ be a $3 \times 3$ symmetric positive definite matrix that is elementwise nonnnegative; we write explicitly
  \begin{equation*}
    D =
    \begin{bmatrix}
      d^2_{x} & d^2_{xy} & d^2_{xz}\\
      d^2_{xy} & d^2_{y} & d^2_{yz}\\
      d^2_{xz} & d^2_{yz} & d^2_{z}
    \end{bmatrix}.
  \end{equation*}
  Let $\bm{\theta}=\half[d^2_{x}, d^2_{y}, d^2_{z}]^T$ and define $M=\bm{\theta1}^T+\bm{1\theta}^T-D$. Then for the statements
  \begin{enumerate}[(i)]
  \item $D$ is positive definite;
  \item $M$ is cnd and nonnegative; and
  \item $d^2(x,y)$ is a squared metric,
  \end{enumerate}
  the following claims hold: $(i) \implies (ii) \Longleftrightarrow (iii)$.
\end{theorem}

\begin{proof}
  \emph{(i) $\implies$ (ii)}: Immediate, as $x^TMx=-x^TDx \le 0$ for any $x \in \reals^3$ such that $x^T\bm{1}=0$. 

  \emph{(ii) $\implies$ (iii)}: First, consider nonnegativty. It suffices to discuss $\alpha$; the others follow similarly. 
  \begin{equation*}
    \alpha = \half(d_x^2+d_y^2)-d_{xy}^2 \ge 0\quad\Leftrightarrow\quad d_{xy}^2 \le \half(d_x^2+d_y^2).
  \end{equation*}
  But $D$ is psd, whereby $d_{xy}^4 \le d_x^2d_y^2$.

  What remains to show is that
  \begin{equation*}
    M = \begin{bmatrix}
      0 & \alpha & \beta\\
      \alpha & 0 & \gamma\\
      \beta & \gamma & 0
    \end{bmatrix}\quad\text{is cnd},
  \end{equation*}
  where $\alpha, \beta, \gamma$ are shorthand for the actual entries of $M$.
  For the vector $x=[-s-t,s,t]$, we have
  \begin{equation}
    \label{metric:eq:5:metric}
    -\half x^TMx = \alpha  s^2+s t (\alpha +\beta -\gamma )+\beta  t^2 \ge 0.
  \end{equation}
  From this inequality we need to deduce that
  \begin{align}
    \label{eq:4}
    \alpha^{1/2} & \le \beta^{1/2} + \gamma^{1/2},\\
    \label{eq:5}
    \beta^{1/2}  & \le \alpha^{1/2} + \gamma^{1/2},\\
    \label{eq:6}
    \gamma^{1/2} &\le \alpha^{1/2} + \beta^{1/2}.
  \end{align}
  Assume without loss of generality that $\gamma$ is the largest. Then, if we can prove that $\gamma^{1/2} \le \alpha^{1/2}+\beta^{1/2}$, the other inequalities follow immediately since $\alpha,\beta\ge0$. To that end, we can equivalently show that
  \begin{equation}
    \label{eq:16}
    \begin{split}
      \gamma &\le \alpha + \beta + 2 \sqrt{\alpha\beta}\quad\Leftrightarrow\quad
      \alpha+\beta-\gamma \ge -2\sqrt{\alpha\beta}.
    \end{split}
  \end{equation}
  Let us see how to deduce \eqref{eq:16} from~\eqref{metric:eq:5:metric}, which says that
  \begin{equation}
    \label{metric:eq:8:metric}
    (\alpha+\beta-\gamma)st \ge -s^2\alpha - t^2\beta.
  \end{equation}
  In particular, let $s^2=\sqrt{\beta}/\sqrt{\alpha}$ and $t^2=\sqrt{\alpha}/\sqrt{\beta}$; this yields $st=1$ and $s^2\alpha+t^2\beta = 2\sqrt{\alpha\beta}$, so that~\eqref{metric:eq:8:metric} reduces to the desired inequality~\eqref{eq:6}.
  

  \vskip4pt
  \noindent\emph{(iii) $\implies$ (ii)}: Let $x=[-s-t,s,t]$ as before. We wish to show that $x^TMx\le 0$; we split this task into two subcases:
  \emph{(a)} $st <0$, and \emph{(b)} $st > 0$.
  
  \emph{Case (a).} Let $\alpha,\beta,\gamma$ be squared distances as before, with $\gamma$ that largest. Then, from inequality~\eqref{eq:6} it follows that
  \begin{equation}
    \label{metric:eq:9:metric}
    \gamma^{1/2} \ge |\alpha^{1/2}-\beta^{1/2}|\quad\implies \gamma \ge \alpha+\beta-2\sqrt{\alpha\beta}.
  \end{equation}
  The only way to violate cnd property of $M$ is to choose $s$ and $t$ such that $x^TMx \ge 0$, or equivalently to show that 
  \begin{equation}
    \label{metric:eq:10:metric}
    s^2\alpha + t^2\beta + st(\alpha+\beta-\gamma) \le 0. 
  \end{equation}
  Since $st < 0$, \eqref{metric:eq:10:metric} turns into
  \begin{equation*}
    \alpha+\beta-\gamma \ge \theta \alpha + \frac1\theta\beta \ge 2\sqrt{\alpha\beta}.
  \end{equation*}
  That is, to break the cnd property of $M$ we need to have
  \begin{equation*}
    \alpha+\beta-2\sqrt{\alpha\beta} \ge \gamma,
  \end{equation*}
  which contradicts~\eqref{metric:eq:9:metric}. 

  \emph{Case (b).} If $st >0$, then from~\eqref{eq:6} it follows that
  \begin{equation*}
    (\alpha+\beta-\gamma)st \ge -2st\sqrt{\alpha\beta} = -2\sqrt{s^2\alpha t^2\beta}
    \ge -s^2\alpha - t^2\beta.
  \end{equation*}
  But this inequality can not contradict the cnd property, as it is just inequality~\eqref{metric:eq:8:metric} analyzed above. Thus, in both cases, we obtain that $M$ must be cnd. 
\end{proof}

\end{document}